\documentclass{article}
\usepackage[utf8]{inputenc}
\usepackage{algorithm}
\usepackage{algorithmic}
\usepackage{amsthm}
\usepackage{pgf, tikz}
\usepackage{cite}
\usepackage{url}
\usetikzlibrary{trees}
\usepackage{amssymb}
\newtheorem{theorem}{Theorem}[section]

\newtheorem{lemma}[theorem]{Lemma}
\usepackage{hyperref}

\theoremstyle{definition}
\newtheorem{definition}{Definition}[section]

\newtheorem{assumption}[theorem]{Assumption}

\title{Power of human--algorithm collaboration in solving combinatorial optimization problems}

\date{July 2021}

\begin{document}

\maketitle

\begin{center}

\author{Tapani Toivonen \\ \href{mailto:tapani.toivonen@uef.fi}{tapani.toivonen@uef.fi} \\
University of Eastern Finland, School of computing}

    
\end{center}

\begin{abstract}
Many combinatorial optimization problems are often considered intractable to solve exactly or by approximation. An example of such problem is \textit{maximum clique} which -- under standard assumptions in complexity theory -- cannot be solved in sub-exponential time or be approximated within polynomial factor efficiently. We show that if a \textit{polynomial time} algorithm can query informative Gaussian priors from an expert $poly(n)$ times, then a class of combinatorial optimization problems can be solved efficiently in \textit{expectation} up to a multiplicative factor $\epsilon$ where $\epsilon$ is arbitrary constant. While our proposed methods are merely theoretical, they cast new light on how to approach solving these problems that have been usually considered intractable.
\end{abstract}

\section{Introduction}
Human-in-the-loop (HITL) AI has gained quite a lot of attention during the past years \cite{1}. HITL AI has emerged to compete with autonomous systems in fields where the interference of a human user is required to solve certain hard problems that might be intractable for an autonomous AI. HITL AI aims to find solutions to problems by involving human to learning process of an AI model. Such interaction can be, for instance, model adjusting, hyperparameter tuning or data processing. As a new competitor of autonomous AI, human--algorithm collaboration has shown promising results in educational data mining and learning analytics \cite{2}. Other promising fields of human--algorithm collaboration include human-in-the-loop optimization \cite{3} and human--robot interaction \cite{4}.

Solving hard combinatorial optimization problems such as \textit{maximum satisfiability} \cite{5}, \textit{maximum clique} \cite{6} or \textit{minimum vertex cover} \cite{7} is usually considered intractable. Consensus in computational complexity theory is that these problems cannot have efficient algorithms that would always yield correct results \cite{8}. That is, no polynomial time algorithms to solve them exist. This further implies that some problems, such as finding the maximum sized cliques or independent sets in graphs, cannot even be approximated well and efficiently in the same time \cite{9} if widely assumed super polynomial lower bounds for the problems hold. 

Solving such hard problems in general has not gained much attention in human--algorithm collaboration theory as main focus has been in AI applications \cite{10} such as learning \cite{11}. In this paper we show that, through human--algorithm collaboration, a combinatorial minimization or maximization problem that has combinatorial complexity of $O(2^n)$ (or as a matter of fact, any $O(poly(n)^{poly(n)})$, can be solved efficiently up to an arbitrary multiplicative factor $\epsilon$ \textit{if} the algorithm can query \textit{Gaussian priors} from a human expert during the execution. Combinatorial complexity means the size of the problem's search space. For instance, in maximum clique, the combinatorial complexity is $O(2^n)$ but in travelling salesman problem, the combinatorial complexity is $O(n!)$. Note that \textit{any} NP optimization problem can be reduced into, say, clique problem, which \textit{has} combinatorial complexity of $O(2^n)$. We further show that the scope of possible problems solvable with the proposed method includes even problems with super-exponential search spaces.

Bayesian optimization is a framework to solve hard function problems and is based on Bayesian statistics. Usually, before the actual Bayesian optimization, an \textit{expert} provides algorithm some information about the problem that is being optimized in a form of a \textit{Gaussian prior}. Gaussian prior reflects one's understanding on shape and smoothness of the objective function. Querying Gaussian priors from human experts is very usual assumption in Bayesian optimization literature \cite{12}. In \textit{fully} Bayesian approach to AI and global optimization, the optimization procedure based on Gaussian processes expects the Gaussian prior $G(\mu, \Sigma)$ as a part of input \cite{13}. That is, the prior is decided by a human expert. What is more, in fully Bayesian case, the prior is expected to be correct. I.e. the function being optimized is a realization of the given Gaussian prior (in Figure \ref{fig:wiener}, 3 realizations of Wiener processes, for instance). In this paper, we make similar assumptions. We expect that the priors are \textit{consistent} and \textit{informative}. These assumptions are quire realistic and very usual, because the expert has access to previous function evaluations, previous problem instances, and multiple prior sampling heuristics such as Maximum Likehood estimation.

This all being said, we stress that we do not argue to solve any notoriously hard problem in theoretical computer science. Instead we give pointers that human--algorithm collaboration can help us solve even  problems considered intractable in the future as the field of human--algorithm collaboration matures.

This paper is organized as follows. First we introduce the concept of Bayesian optimization. Second we show how to reduce a class of combinatorial optimization problem instances to a \textit{univariate} finite domain function, which then can be approximated by our human--algorithm collaboration Bayesian optimization algorithm that we introduce thereafter. Finally we conclude our research and give some future directions.

\begin{figure}
    \centering
    \includegraphics[width=\textwidth]{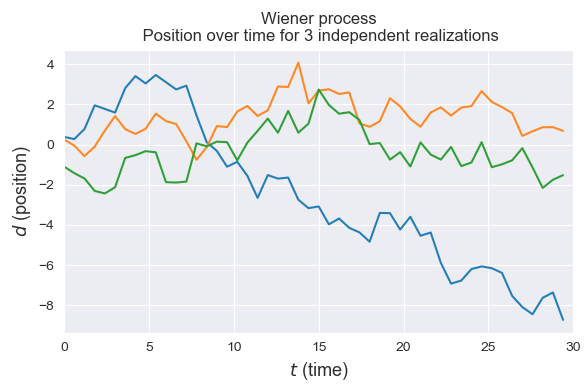}
    \caption{3 realizations of Wiener process}
    \label{fig:wiener}
\end{figure}

\section{Bayesian optimization}
Global optimization (GO) aims to find optimal value(s) of functions called \textit{objective functions} either in finite or bounded domains \cite{14}. The optimal values, depending on the context, can be the global maximal or global minimal values. In general, global optimization is intractable -- the number of function evaluations increases exponentially in the problem dimensions and exponentially in domain size \cite{15}. The GO problems for continuous functions and functions with finite domains respectively are

\begin{equation}
     \textrm{maximize} \; F(x), x \in [0, 1]^d, d \in \mathbb{N}
\end{equation}

\begin{equation}
    \textrm{maximize} \; F(x), x \in N, N \in \mathbb{N}
\end{equation}

Note that any minimization problem can be reduced into a maximization problem by $F := -F$.

Bayesian optimization (BO) is a technique used in GO to search globally optimal values in continuous, combinatorial, and discrete domains \cite{a7}. BO has wide range of applications and during the recent years, its usage in optimizing hard black box functions has increased \cite{a8}. BO is often used in low data regimes where where evaluation of the objective function is costly or not otherwise efficient \cite{a9}. These problems include, robot navigation and planning \cite{a10}, tuning the hyperparameters of deep learning models \cite{a11}, predicting earthquake intensities \cite{a12}, finding optimal stock values in stock markets \cite{a13} and much much more \cite{a14}.

The advantage of the BO is that BO can optimize any set of black box functions. That is, if one can only access function values and not, say, gradient information, then BO can be very efficient \cite{a8}. Also, BO does not usually make any assumptions on the function it optimizes unlike multiple state-of-the-art optimization algorithms \cite{a7}. These algorithms expect the objective function to be Lipschitz continuous \cite{a16}, unimodal \cite{a17} or having a certain shape near to global optimizer \cite{a18}. \textit{Fully} Bayesian optimization, however, assumes that the user has some knowledge or expectation on the shape of the objective function, which realizes as (Gaussian) \textit{prior}. The main difference between BO and traditional GO is that in the BO, one is not assumed to provide, say, differentiable function -- one just has to know to some extent what kind of the function is.

Bayesian optimization is mainly based on Gaussion processes \cite{refpaper}. Gaussian process is a stochastic process such that every finite collection of random variables has a multivariate normal distribution. Gaussian process is completely defined by its convariance function $\Sigma$ (or covariance matrix in finite domains) and its mean function $\mu$ \cite{a7}. That is, $G(\mu, \Sigma)$.

In this paper it is assumed that for \textit{any} function $F$ considered

\begin{equation}
    F \sim G(\mu, \Sigma)
\end{equation}

, and that $\mu$ is some constant, say 0 for centered Gaussian process. This is very usual assumption made in BO literature.

In BO, one places a (Gaussian) \textit{prior} on the unknown and possible non-convex function which reflects one's understanding of the function -- whether the function is continuous, differentiable, smooth, unimodal and so forth \cite{a7}. It is usually assumed that $\mu$ is 0 and hence, the prior is defined only by its covariance function (or covariance matrix) $\Sigma$. After the prior is set over the unknown function, an algorithm suggests points where the optimal values would lie. Based on these suggestions and function evaluations, the algorithm updates the prior to \textit{posterior} and uses the posterior the suggest even better points where the possible optimum would be located. In many settings, the BO finds the global optimum much more faster than, say, brute force search \cite{a8}. The suggestions and how they are derived from the prior and posterior are defined by an acquisition function \cite{a7}. Multiple different possibilities for acquisition function exists. These include upper confidence bound \cite{a20}, expected improvement \cite{a21}, and Thompson sampling \cite{a22} to name few. A visualization of BO with posterior distribution of functions, sample points, and an acquisition function can be seen in Figure \ref{fig:BO}.

\begin{figure}
    \centering
    \includegraphics[width=0.7\textwidth]{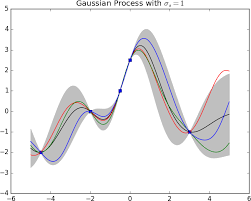}
    \caption{Bayesian optimization}
    \label{fig:BO}
\end{figure}

Multiple studies have shown that Bayesian optimization converges to global optimum in reasonable time \cite{a7} and \cite{a8}. \cite{a14} showed exponential convergence rate of Bayesian optimization to expected global optimum (that is, the expected highest or lowest value of the function depending on the context) if the function near to global optimum has a certain shape. \cite{16} showed similar results to more restricted version of Bayesian optimization where the function is a realization of a Gaussian process called Brownian motion (also known as Wiener process).

While it is tempting to praise Bayesian optimization for its ability to find optimal values of the structures efficiently, one of the downsides of the approach is that the locally optimal values found and globally optimal values are based on how well the prior of the Gaussian process is defined. If the black box function is a realization of the prior, then BO works surprisingly well. On the other hand, a poor choice of a prior can result the algorithm to not converge at all. In Bayesian optimization literature \cite{a7} and \cite{a8}, it is usually assumed that an expert, who has a domain knowledge in the context of the objective function provides a prior from which the objective function is realization.

The \textit{regret} \cite{a9}, which is used to give an idea of the convergence rate of a BO algorithm, is based on the expectations. That is, ratio between \textit{expected} global optimum and \textit{expected} best value found the function. The regret is defined as

\begin{equation}
    r_T = \mathbf{E}[F_{sup}] - \mathbf{E}[Y]
\end{equation}

, where $F_{sup}$ is the global optimum and $Y$ is the best value found by a BO optimization algorithm. The expectation in the regret means that on average among all functions, the Bayesian optimization algorithm will return a value that is $r_T$ close to an expected optimum. The expected values are bounded by the prior and posterior distributions of Gaussian process. In \cite{refpaper} the regret was defined as a multiplicative factor of expected best value found and expected global optima

\begin{equation}
    r_T := normregret = \frac{\mathbf{E}[F_{sup}] - \mathbf{E}[Y]}{\mathbf{E}[F_{sup}]}
\end{equation}

, which is closely related to an approximation ratio in approximation. In this paper, we complement the previous results of \cite{refpaper} to obtain (super-)exponential convergence rate for BO. We restrict ourselves to the objective functions that can be \textit{changed} while preserving the relative scale of the function values in the original objective function. We show that any combinatorial optimization problem with combinatorial complexity of $O(poly(n)^{poly(n)})$ can be reduced to such function. Later, we show how to derive the (super-)exponential convergence rate with our human--algorithm collaboration procedure.

\section{Combinatorial problems as univariate finite domain functions}

In this section, we show that any problem with a \textit{combinatorial} complexity of $O(2^n)$ can be reduced to a \textit{univariate} function with a domain of size $\Omega (2^n)$. The reduction is quite general and can be applied also for problems with combinatorial complexity of $O(poly(n)^{poly(n)})$: one just has to consider different branching factor $m$.

The reduction algorithm (seen in Algorithm 1) assumes that the problem instance can be viewed as a decision tree (see Figure 2) where a left child arc of a decision node implies that the decision variable at the node is assigned with 0 and right child arc implies 1 assignment (or vice versa). The function that the reduction produces, takes value in a domain $D = [D_0, D_1]$ and uses the value to find an assignment for the combinatorial problem instance. Finally the function evaluates the combinatorial problem with the assignment derived from the input value.

The algorithm starts by ordering the decision variables randomly. Second the algorithm creates an objective function that accepts a value from a finite domain $D = [D_0, D_1]$. Based on the value passed to the function, the function always divides the remaining domain in the two equally sized halves and selects the half that contains the value. At each selection of the half, the assignment is updated with the corresponding decision variable value based on which of the halves was selected. When the domain has been completely split so that no more intervals can be selected, the original combinatorial problem instance is evaluated with the assignment and the value of the evaluation (number of vertices in a clique, number of clauses satisfied, or number of vertices in a vertex-cover and so forth) is scaled with the scaling parameter and eventually returned. The algorithm can be seen in Algorithm \ref{algo:1} and the geometric presentation of the combinatorial problem in Figure \ref{fig:dpll}.

If Algorithm \ref{algo:1} is called with a partial assignment (that is, some variables have already been assigned with values), then the algorithm starts by randomly ordering the variables that have not been assigned to any values. Then the algorithm produces the function for the remaining variables and uses the partial assignment as a basis of the final evaluation.

\begin{algorithm}
\caption{Combinatorial problem to finite domain univariate function}
\label{Sample}
\begin{algorithmic}[1]
\STATE \textbf{input: X (problem instance), V (variables), scale (scale of function values), P (partial assignment), $D_0$, $D_1$}
\STATE randomly sort V (if P contains values then discard assigned variables from V)
\STATE create black-box function with parameters x (point to evaluate), $[D_0, D_1]$:
\STATE current := $[D_0, D_1]$
\STATE $i$ := 0
\WHILE{current $\neq$ single value}
    \STATE a := first \textit{half} of current
    \STATE b := second \textit{half} of current
    \STATE current := a \textbf{if} x $\in$ a \textbf{else} b
    \STATE $V_i$ := \{1 or 0\} based on current (whether a or b was selected)
    \STATE $i$ := $i$ + 1
\ENDWHILE
\STATE assignment := evaluate X with assignment of V and P
\STATE y := value of assignment (value of maximization / minimization problem)
\STATE scale y based on \textit{scale}
\STATE \textbf{return} black box function (lines 3 - 15)
\end{algorithmic}
\label{algo:1}
\end{algorithm}

\subsection{Analysis of Algorithm \ref{algo:1}}

In a combinatorial problem, whose combinatorial complexity is $O(2^n)$, there exist $2^n$ different variable assignments. Algorithm \ref{algo:1} runs $O(n)$ steps and at each step, the algorithm splits the remaining domain in two halves and selects another half as a new domain. The assignment for the combinatorial problem is updated based on the order of the variables, the depth at which the algorithm currently operates and whether the value of input $x$ belongs to first or the second half of the split interval (other half indicates, say, 0 value and other half indicates 1 value). This way, the algorithm assigns each value of input $x \in D, D = [D_0, D_1]$ ($[1, 2^n]$, for instance) with a different assignment, which is used to evaluate the combinatorial problem at the end of the algorithm. Scaling parameter is used for scaling the returned value if the new scale is required -- in order to produce functions from a specific Gaussian process prior. Scaling factor can also be used to group certain optimization results for the same function values. Say, 1 or 2 clauses satisfied by an assignment yields function value 1.

Random ordering of the variables from the same combinatorial optimization problem instance produces always same function values that lie in the different positions in the function domain --- even if some variables are fixed as in partial assignment's case.

\begin{figure}
    \centering

\tikzstyle{level 1}=[level distance=2.5cm, sibling distance=3.5cm]
\tikzstyle{level 2}=[level distance=2.5cm, sibling distance=2cm]

\tikzstyle{bag} = [text width=4em, text centered]
\tikzstyle{end} = [circle, minimum width=3pt,fill, inner sep=0pt]

\tikzset{
  mynode/.style={fill,circle,inner sep=1pt,outer sep=0pt}
}

\begin{tikzpicture}[grow=down, sloped]
\node[bag] {$v_1$}
    child {
        node[bag] {$v_2$}        
            child {
                node[end, label=right:
                    {}] {}
                edge from parent
                node[above] {$-$}
            }
            child {
                node[end, label=right:
                    {}] {}
                edge from parent
                node[above] {$+$}
            }
            edge from parent 
            node[above] {$-$}
    }
    child {
        node[bag] {$v_2$}        
        child {
                node[end, label=right:
                    {}] {}
                edge from parent
                node[above] {$-$}
            }
            child {
                node[end, label=right:
                    {}] {}
                edge from parent
                node[above] {$+$}
            }
        edge from parent         
            node[above] {$+$}
    };
    \node[draw,text width=6cm] at (0,1) {
        V = $(v_1, v_2, ..., v_n)$
    };
    
    \draw[black,thick,-] (-3,-6) -- (3,-6)
    node[pos=0,mynode,fill=black,label=above:\textcolor{black}{1}]{}
    node[pos=1,mynode,fill=black,text=green,label=above:\textcolor{black}{$2^n$}]{};

\end{tikzpicture} 
    \caption{Geometric binary tree ($m = 2$) presentation of a combinatorial problem}
    \label{fig:dpll}
\end{figure}
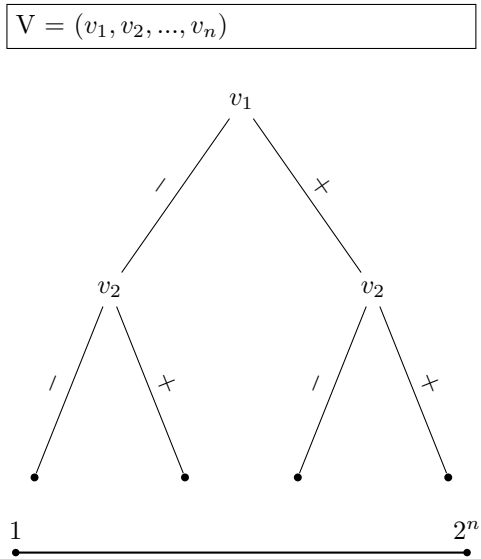

\section{Human--algorithm collaboration in Bayesian optimization}

In this section we refer to a recent paper of \cite{refpaper} where the authors showed that any ratio between expected global optimum and expected optimum found by their Bayesian optimization (BO) algorithms (UCB2 or EI2) in functions with finite domains. We complement their results to show that for a specific type of problems, their results can be extended to a polynomial time human--algorithm collaboration procedure. We stress that the algorithm is merely theoretical and is required to query Gaussian priors from an external expert multiple times.

The key ingredient of their work that we use is that -- instead of standard simple regret used in Bayesian optimization -- they provided proof for expected approximation ratio -- the multiplicative ratio between the expected best value found by the function and the expected global optimum (normregret in equation 5). Also, their bound of the regret ratio is tight up an arbitrary constant. That is, their UCB2 or EI2 are optimal to the worst case bound. This means that the upper bound equals to lower bound in their algorithms up to some constant $\epsilon$. Furthermore, their regret ratio is not asymptotic -- which is crucial in our analysis. The upper and lower bounds for their algorithms are respectively

\begin{equation}
        normregret \leq 1 - (1 - T^{\frac{1}{2\pi}}) \cdot \frac{\sqrt{\log_2(T) - \log_2(3\log_2^{\frac{3}{2}}(T))}}{\sqrt{\log_2(N)}}
    \end{equation}

and

\begin{equation}
    normregret \geq 1 - \frac{\sqrt{\log_2(T)}}{\sqrt{\log_2(N)}} - \epsilon
\end{equation}

, where $T$ is the number of function evaluations, $N$ is the domain size, and $\epsilon$ is arbitrary constant depending on $T$ ($\epsilon$ tends to approach 0 as $T$ increases).

As such, their work cannot be extended to polynomial time algorithm to achieve a constant factor approximation in polynomial time due to their bounds are only \textit{expectations} (as in standard BO). Hence, to obtain the strict ratio provided by their algorithm, one would have to evaluate arbitrary number of different functions. We overcome this constraint by narrowing down the scope of the functions to type of functions where the function can be changed almost arbitrary many times but still depend on the same problem instance (see Algorithm \ref{algo:1}).

Without a loss of generality, we assume, once again, that the combinatorial problems here have combinatorial complexity of $O(2^n)$. In Algorithm \ref{algo:1} we showed that any combinatorial problem with combinatorial complexity of $O(2^n)$, can be reduced to a black-box univariate function with a finite domain. By running the Algorithm 1 multiple times for a single instance with different scaling factors, the algorithm produces different functions from the same problem instance. This is because the algorithm always randomly sorts either all or a subset of the variables (line 2 in Algorithm \ref{algo:1}). We use that and \textit{McDiarmid’s inequality} \cite{mcd} to obtain with high probability a convergence to expected values. McDiarmid’s inequality is defined as

\begin{equation}
    |F(x_1, x_2, ..., x_m) - F(x_1, x_2, ..., x_i', ..., x_m)| \leq c_i
\end{equation}

, where a value change of a single random variable can cause a function value to change at most $c_i$, which resembles Lipschitz condition. 

and the actual inequality

\begin{equation}
    P[|F(x_1, x_2, ..., x_m) - E[F(x_1, x_2, ..., x_m)]| \geq t] \leq exp(\frac{-2t^2}{\Sigma_{i=0}^m(c_i^2)})
\end{equation}



, where $t$ is some constant for the difference between function of random variables and expected value of the function with random variable. McDiarmid’s inequality is a concentration inequality and states that the value of a function for random variables converges to its expected value exponentially in the number of random variables if a change in one of it's parameters changes the function value only by a samll fraction ($c_i$). In our case, we run our reduction algorithm (Algorithm 1) $S$ times, use either UCB2 or EI2 \cite{refpaper} each time and finally output the mean of the found values by UCB2 or EI2. By McDiarmid’s inequality, the mean value converges exponentially fast to the bounds promised by UCB2 or EI2. We then branch to the optimal regions of the domain where the optimal function values are likely to reside. Hence, our algorithm follows quite popular geometric branch and bound framework used in multiple global optimization algorithms \cite{16} and \cite{18}. 

In \cite{16}, the search space is shrunk after function evaluations and the intervals where the global optimizer is not likely to reside are discarded. \cite{16} and \cite{18} cannot be used as such to solve problems with exponential sized domains because constants in their algorithm depend on the domain size and the shape of the function near to optimizer(s).

Because the prior given to the Gaussian process upper bounds the expected optimum and the expected best value found by UCB2 and EI2, at every expansion and new sample from Algorithm \ref{algo:1}, we query a new prior from an external expert. We assume that the priors queried from the expert do not contradict each others and are informative. That is, there are no conflicting choices of, say, hyperparameters, among the priors regarding the same fractions of the search space, and that the expert can gain insights of the functions. 

Querying informative priors from an expert is realistic assumption -- and usually made in BO literature -- since the expert has an access to multiple sampling heuristics, such as Maximum likehood estimation and previous function evaluations and \textit{domain knowledge}. This differentiates querying priors from an expert to standard \textit{oracle} queries in theory of computation where the oracle is unrealistic non-deterministic entity who is always correct. Our extension to \cite{refpaper} algorithm can be seen in Algorithm \ref{algo:2}.

Before going into details we introduce some definitions and an assumption.

\begin{definition}[Cell]
    Interval produced when a domain or a part of the domain is divided into two halves. One division produces two new cells.
\end{definition}

\begin{definition}[Expansion of cell]
    Division of a cell in two intervals equal size.
\end{definition}

\begin{definition}[Upper bound]
    Expected maximum of a cell. Upper bound is calculated from the size of a cell, mean maximum found by EI2 or UCB2 from samples of Algorithm \ref{algo:1}, and number of iterations used by EI2 or UCB2 ($X$): 
    
    \begin{equation}
        ub := \frac{val}{\frac{\sqrt{\log_2(T)}}{\sqrt{\log_2(N)}}}
    \end{equation}
    
    , where $val$ is the \textit{mean} of values found by EI2 or UCB2 in S runs from a cell, $T = X$, and $N$ is the size of the cell. Ub is calculated for each cell as there were no approximation error.
    
\end{definition}

\begin{definition}[$\epsilon$ optimal solution]
    An expected local optimum $\epsilon$ close to a \textit{expected} global optimum.
\end{definition}

\begin{definition}[Optimal cell]
    Cell that contains $\epsilon$ optimal solution.
\end{definition}

\begin{assumption}
    The expert's priors do not contradict each others at any point in the same fractions of the search space, and priors are \textit{informative}.
\end{assumption}

\begin{algorithm}
\caption{Bayesian optimization with expert knowledge}
\label{build}
\begin{algorithmic}[1]

\STATE \textbf{inputs: D (domain), $T$, $S$, $X$, $V$}

\STATE current := $[1, D]$

\STATE $t := 0$
\STATE $i := 0$
\STATE $j := 0$

\STATE fix order of $Vars := V$

\WHILE{$t \leq \log_2(T)$}

    \STATE a := first \textit{half} of current
    \STATE b := second \textit{half} of current
    \STATE remove current cell
    \STATE add a to cells
    \STATE add b to cells
    
    \STATE $s := 0$
    
    \WHILE{$s \leq \log_2(S)$}
        
        \STATE re-sample a and b from Algorithm 1, use absolute position of a and b (in original domain) to derive a partial assignment for Algorithm \ref{algo:1} and $D_0, D_1$.
        \STATE derive covariance matrices ($\Sigma_i, \Sigma_j$), and means ($\mu_i, \mu_j$) for a and b using query to \textbf{expert}
        \STATE solve UCB2 or EI2 for a, use $T := X$  iterations, covariance matrix derived $\Sigma_i$, and $\mu_i$
        \STATE solve UCB2 or EI2 for b, use $T := X$  iterations, covariance matrix derived $\Sigma_j$, and $\mu_j$
        
        \STATE $s := s + 1$
        \STATE $i := i + 1$
        \STATE $j := j + 1$
        
    \ENDWHILE
    
    \STATE retain original order of variables for $a$ from $Vars$
    \STATE retain original order of variables for $b$ from $Vars$
    
    \STATE deduce $ub$ of $a$ from its found UCB2 or EI2 values in S samples from Algorithm 1, size of $a$, and X
    \STATE deduce $ub$ of $b$ from its found UCB2 or EI2 values in S samples from Algorithm 1, size of $b$, and X
    
    \STATE current := \textit{argmax} of cells (the cell with highest $ub$)
    
    \STATE $t := t + 1$
    
\ENDWHILE
\STATE \textbf{return} max value of any cell found by UCB2 or EI2
\end{algorithmic}
\label{algo:2}
\end{algorithm}

In Algorithm \ref{algo:2} we show the pseudocode of the algorithm. If the priors given by the expert are informative and consistent, then Algorithm \ref{algo:2} converges to $\epsilon$ optimal solution in $O(S \cdot A \cdot K \cdot C \cdot \log_2(D))$ for some $C$ depending on the problem instance, $S$ and $X$; and $K$ (the number of function values $\epsilon$ close to global optima). This is because the Algorithm 2 with high probability expands the $\epsilon$ optimal cell in two equal sized halves -- the rate of convergence is exponential. In the following the prove our claims.

\begin{lemma}
     With probability in $S$, Algorithm 2 expands optimal cell at any time $t_i \in T, 0 \leq i \leq T$
\end{lemma}

\begin{proof}
    By McDiarmid’s inequality,
    
    \begin{equation}
        |F(x_1, x_2, ..., x_m) - F(x_1, x_2, ..., x_i', ..., x_m)| \leq c_i
    \end{equation}

    and

    \begin{equation}
        P[|F(x_1, x_2, ..., x_m) - E[F(x_1, x_2, ..., x_m)]| \geq t] \leq exp(\frac{-2t^2}{\Sigma_{i=0}^m(c_i^2)})
    \end{equation}
    
    , the Algorithm 2 can be modified so that it passes all $S$ random permutations from Algorithm 1 to a black box function that optimizes permutations with EI2 or UCB2 and outputs the mean of the results. Hence, $c_i$ is bounded by $\frac{1}{m} \Sigma_{i=0}^m (b - a)$ where $b$ and $a$ are upper bound and lower bound of the functions values respectively. We have hence,
    
    \begin{equation}
        P[|F(x_1, x_2, ..., x_m) - E[F(x_1, x_2, ..., x_m)]| \geq t] \leq exp(\frac{-2m^2t^2}{(b - a)^2})
    \end{equation}
    
    , This implies that the value of $S$ UCB2 or EI2 runs converges exponentially fast to its expected value in $m$ -- the number of random permutations. 
    
    The expected approximation ratio in \cite{refpaper} is given as
    
    \begin{equation}
        1 - (1 - T^{\frac{1}{2\pi}}) \cdot \frac{\sqrt{\log_2(T) - \log_2(3\log_2^{\frac{3}{2}}(T))}}{\sqrt{\log_2(N)}}
    \end{equation}
    
    , which is,
    
    \begin{equation}
        \Omega((1 - T^{\frac{1}{2\pi}}) \cdot \frac{\sqrt{\log_2(T) - \log_2(3\log_2^{\frac{3}{2}}(T))}}{\sqrt{\log_2(N)}})
    \end{equation}
    
    This is upper bounded by
    
    \begin{equation}
        o(\frac{\sqrt{\log_2(T)}}{\sqrt{\log_2(N)}})
    \end{equation}
    
    By McDiarmid’s inequality, the mean approximation ratio tends in between these bounds in $m := S$.
    
    The cell with the highest upper bound is always expanded. This, the upper, and the lower bound imply that the cell that does not contain the expected global optimizer might be expanded instead of the cell with the optimizer -- some cell might have lower approximation error and only slightly smaller expected optimum, and hence, larger $ub$.
    
    This is limited by
    
    \begin{equation}
        E[F_{sup}] \cdot (1 - T^{\frac{1}{2\pi}}) \cdot \frac{\sqrt{\log_2(T) - \log_2(3\log_2^{\frac{3}{2}}(T))}}{\sqrt{\log_2(N)}} > E[F_{local}] \cdot \frac{\sqrt{\log_2(T)}}{\sqrt{\log_2(N)}}
    \end{equation}
    
    , which is
    
    \begin{equation}
         E[F_{sup}] > E[F_{local}] \cdot \frac{\sqrt[2\pi]{T} \cdot \sqrt{\log_2(T)}}{(\sqrt[2\pi]{T} - 1)\sqrt{\log_2(T)-\log_2(3\log_2^{\frac{3}{2}}(T))}} + \epsilon
    \end{equation}
    
    , where $F_{local}$ is a maximum of a cell other than expected global optima, and $\epsilon$ is a small constant from McDiarmid’s inequality's convergence error. Equation (18) gives a limit for $\epsilon$ optimal solution.
    
    This implies that even a cell without any approximation error, if the found value is less than (18) factor from expected global maximum, then the cell will not be expanded because it will be \textit{dominated} by at least the cell with the expected global optimizer.
    
    These prove that -- with probability in $m := S$ -- only optimal cells will be expanded.
    
\end{proof}

\begin{lemma}
    Algorithm \ref{algo:2} will, with very high probability, find at least $\epsilon$ optimal solution to any optimization problem derived from some Gaussian process and Algorithm \ref{algo:1}.
\end{lemma}

\begin{proof}
    When running Algorithm 2 sufficiently many iterations, the values found by UCB2 and EI2 increase, as cells' size decrease \cite{refpaper}, and tend to cell's exåected optimum.
     
     The probability $x$ of selecting \textit{always} optimal cell is propositional to By McDiarmid’s inequality, $m := S$, but decreases exponentially to a depth of the tree and tend to 0:
    
    \begin{equation}
        P[x|h] := (1 - (e^{-\frac{2m^2t^2}{(b - a)^2}}))^h
    \end{equation}
    
    , where $h$ is the maximum depth of the tree. On the other hand because
    
    \begin{equation}
        P[x] := 1 - (e^{-\frac{2m^2t^2}{(b - a)^2}}))
    \end{equation}
    
    $P[x]$ increases exponentially in $m$ and tend to 1, the expected required samples from Algorithm 1 to obtain $P[x|h] > 0.5$ is bounded polynomially in $m := O(log_2(D)) = O(h)$.
\end{proof}

\begin{lemma}
    Expected run time of Algorithm 2 to find at least $\epsilon$ optimal solution is $O(S \cdot A \cdot K \cdot C \cdot log_2(D))$ (the time complexity of UCB2 or EI2 omitted here) if the function has $K$ number of $\epsilon$ optimal solutions, where $D$ is the domain size, $A$ is the number of variables in the combinatorial optimization problem, and $C$ depends on function values and $X$.
\end{lemma}

\begin{proof}
    By lemma 4.3, the algorithm always expands the optimal cell with high probability in $S$. Because the expansion procedure divides the cell always in two equally sized halves and because $\epsilon$ optimal cell is expanded with high probability, the convergence rate is $O(\frac{1}{2^N})$. Fix $N = log_2(D)$, we have $O(S \cdot A \cdot K \cdot C \cdot log_2(D))$. $C$ depends on when the Algorithm \ref{algo:2} is able to distinguish the expected global optimum from similar values of the objective function. This further depends on parameter $X$ for UCB2 or EI2. This completes our proof.
\end{proof}

Of course, the approach provided here is merely theoretical. First assumption is that all functions are indeed realization of the Gaussian processes (which might not be the case) and secondly, as the domain size increases, the more priors (even in the best possible case) the expert has to provide, which is not feasible even for even moderate sized problems in real life. The other disadvantage of the algorithm is that some covariance functions, UCB2 or EI2 may produce better results and have less strict bounds because the bounds are general and not prior dependent. This of course, leads to a situation where for some priors, the Algorithm 1 spends more time in searching for sub-optimal regions of the search space. A sample run of the search space shrinking can be seen in Figure \ref{fig:dpll2}. The term $\epsilon$ is very very small. For instance, fixing $T := X = 10^7$ yields less than $2\%$ approximation error for large domains. In equation (18) we showed that the approximation error in propositional only to $X$.

\begin{figure}
    \centering

\tikzstyle{level 1}=[level distance=2cm, sibling distance=2.8cm]
\tikzstyle{level 2}=[level distance=2cm, sibling distance=1.4cm]
\tikzstyle{level 3}=[level distance=1.2cm, sibling distance=0.9cm]

\tikzstyle{bag} = [text width=4em, text centered]
\tikzstyle{end} = [circle, minimum width=3pt,fill, inner sep=0pt]

\tikzset{
  mynode/.style={fill,circle,inner sep=1pt,outer sep=0pt}
}

\begin{tikzpicture}[grow=down, sloped]
\node[bag] {$v_1$}
    child {
        node[bag] {$v_2$}        
            child {
                node[] {$v_3$}
                edge from parent
                node[above] {$-$}
            }
            child {
                node[] {$v_3$}
                child {
                node[] {$v_4$}
                edge from parent
                node[above] {$-$}
            }
            child {
                node[] {$v_4$}
                edge from parent
                node[above] {$+$}
            }
                edge from parent
                node[above] {$+$}
            }
            edge from parent 
            node[above] {$-$}
    }
    child {
        node[bag] {$v_2$}
        edge from parent         
            node[above] {$+$}
    };
    \node[draw,text width=6cm] at (0,1) {
        V = $(v_1, v_2, ..., v_n)$
    };
    
    \draw[black,thick,-] (-3,-6) -- (3,-6)
    node[pos=0,mynode,fill=black,label=above:\textcolor{black}{1}]{}
    node[pos=1,mynode,fill=black,text=green,label=above:\textcolor{black}{$2^n$}]{};
    
     \draw[black,dotted,-] (0,0) -- (0,-6);

     \draw[black,dotted,-] (-1.4,0) -- (-1.4,-6);
     \draw[black,dotted,-] (3,0) -- (3,-6);
     \draw[black,dotted,-] (-3,0) -- (-3,-6);
     \draw[black,dotted,-] (-0.7,0) -- (-0.7,-6);
     
\end{tikzpicture} 
    \caption{Algorithm \ref{algo:2} expands the cell with the highest upper bound ($-v_1$, $+v_2$) and evaluates the upper bounds of cells $(-v_1, +v_2, +v_3)$ and $(-v_1, +v_2, -v_3)$ (upper bounds of cells $(+v_1$), and $(-v_1, -v_2)$ are known)}
    \label{fig:dpll2}
\end{figure}
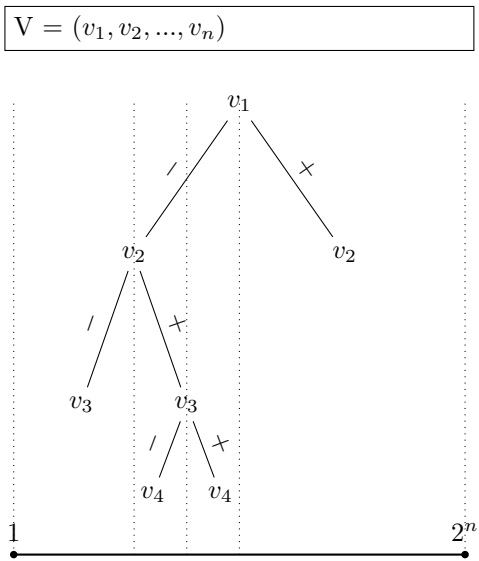

\subsection{Wiener processes revisited}
Wiener process $W$ is a stochastic Gaussian process defined by independent and Gaussian increments

\begin{equation}
    W_{t + u} - W_{t} \sim N(0, u)
\end{equation}

and 

\begin{equation}
    W_0 = 0
\end{equation}

The covariance function is defined in $W$ for $W_s$ and $W_t$, $t \geq s$ as 

\begin{equation}
    \textrm{min}\{s, t\}
\end{equation}

, for Wiener process with \textit{unit} variance

Because the only hyperparameter for the Wiener process is the variance, which can be considered as unit variance by scaling factor in Algorithm 1, a selection of Wiener processes as statistical model in Algorithm 2 requires no expert interference.

UCB2 or EI2 algorithm might not offer a tight bounds for (discrete) Wiener process. However, by lower bound theory there must exist an algorithm where the non-asymptotic normregret upper bound matches to normregret lower bound at least up to a constant factor. One of such algorithms is \textit{Pure random search} for standard -- continuous Wiener process \cite{calvin}. For a standard Wiener process, one can easily modify Algorithm 1 to reduce a combinatorial optimization problem into continuous domain function as $X = [1, D]^d$, $d = 1$.

This implies that by considering Gaussian processes where the only hyperparameter of the model is a scaling factor, one can use the Algorithms 1 and 2 for the traditional optimization.

In Wiener process' case, the approximation ratio is based on the function value distribution under the Wiener measure.

\section{Conclusions}

In this paper we have shown that with a type of human-in-the-loop Bayesian optimization one can solve any combinatorial optimization problems (either minimization or maximization) in polynomial time as long as the combinatorial complexity of the problem is $O(2^n)$ and we can assume that expert can correctly provide $O(n)$ Gaussian process priors. While our proposed algorithm is merely theoretical and has only little practical interest, we assume that our approach could be used to gain new insights on how to tackle the hardest NP-hard combinatorial problems in the future.

As it is usually expected, most of combinatorial optimization problems might be intractable. Our research casts new light how to avoid the common pitfalls faced while solving these problems: to involve human expertise in the optimization process. In our approach, the human expert is not required to know where the optimizers lie in the search space. Instead, the expert is only assumed to provide a Gaussian prior when the expert is queried from. The prior captures expert's opinions on how drastically the nearby values might change in the function and whether the changes are periodic and so forth.

The expert knowledge can be supported with previously calculated values of the objective function and, for instance, Maximum likehood estimation \cite{16}. \cite{16} gave bounds on \textit{sample sizes} of random variables (in our case, the function values) to derive a certain probability in optimization of a likehood function for covariance hyperparameters. Because the bounds do not depend on the domain size of the function, fixing a sample size and a covariance function for covariance matrix $\Sigma$ yields a constant time estimation of a prior. This can be combined with the expert knowledge.

Our results indicate that in the future, in order to understand the fundamental limits of our tools (whether it is a combinatorial optimization algorithm or an AI system) a human interference might be required. The scope of possible comibinatorial problems that can be solved efficiently in theory through our human--algorithm collaboration procedure include enormous amount of combinatorial problems. Such problems include \textit{protein structure prediction, finding maximum sized cliques, model checking, automated theorem proving} (basically every problem in NPO and even beyond to problems with super-exponential combinatorial complexity).

\end{document}